%% file: paper.tex
\documentclass[sigconf]{acmart}

\usepackage[english]{babel}
\usepackage{blindtext}

\renewcommand\footnotetextcopyrightpermission[1]{} 
\setcopyright{none}

\usepackage{algorithm}
\usepackage{algorithmic}

\usepackage{tikz}
\usetikzlibrary{calc}
\usetikzlibrary{patterns}
\tikzstyle{mybox} = [draw=black, fill=white,  thick,
    rectangle, inner sep=10pt, inner ysep=20pt]
\tikzstyle{mybox} = [draw=black, fill=white,  thick,
    rectangle, inner sep=2pt, inner ysep=2pt]
\usetikzlibrary{shapes,snakes}

\usepackage{caption}
\usepackage{subcaption}

\acmDOI{}

\acmISBN{}

\acmPrice{}

\begin{document}
\title{On distributed algorithms for minimum dominating set problem and beyond}



 \author{Sharareh Alipour}
 \affiliation{%
   \institution{School of computer science, Institute for Research in Fundamental Sciences (IPM)}
}
 \email{alipour@ipm.ir}

  \author{Mohammadhadi Salari}
 \affiliation{%
   \institution{Simon Fraser University }
}
 \email{msalari@sfu.ca }
\renewcommand{\shortauthors}{Alipour and Salari}
\renewcommand{\shorttitle}{Distributed algorithms for MDS problem}

\begin{abstract}
 In this paper, we study the minimum dominating set (MDS) problem and the minimum total dominating set MTDS) problem which have many applications in real world. We propose a new idea to compute approximate MDS and MTDS. Next, we give an upper bound on the size of MDS of a graph. We also present a distributed randomized algorithm that produces a (total) dominating subset of a given graph whose expected size equals the upper bound.
Next, we give fast distributed algorithms for computing approximated solutions for the MDS and MTDS problems using our theoretical results. 

The MDS problem arises in diverse areas, for example in social networks, wireless networks, robotics, and etc. Most often, we need to compute MDS in a distributed or parallel model. So we implement our algorithm on massive networks and compare our results with the state of the art algorithms to show the efficiency of our proposed algorithms in practice.
We also show how to extend our idea to propose algorithms for solving $k$-dominating set problem and set cover problem. Our algorithms can also handle the case where the network is dynamic or in the case where we have constraints in choosing the elements of MDS.

\end{abstract}

\maketitle

\input{body}

\bibliographystyle{ACM-Reference-Format}
\bibliography{sample-bibliography}

\end{document}

%% file: body.tex
\section{Introduction}
This paper deals with fast distributed algorithms to compute dominating set and total dominating set of graphs. 
Given a graph $G=(V,E)$ with the vertex set $V$ and the edge set $E$, we show the set of adjacent vertices to a vertex $v$, neighbors of $v$, by $N(v)$.
A set $S \subseteq V$ is a dominating set of $G$ if each node $v \in V$ is either in $S$ or has a neighbor in $S$. Also, $S \subseteq V$ is a total dominating set of $G$ if each node $v\in V$ has a neighbor in $S$. Let $\gamma(G)$ and $\gamma_t(G)$ be the size of a minimum dominating set (MDS) and a minimum total dominating set (MTDS) of a graph $G$ without isolated vertex, respectively. It is easy to prove that 
$$\gamma(G)\leq \gamma_t(G) \leq 2\gamma(G).$$

An extension of MDS problem is minimum $k$-distance dominating set problem where the goal is to choose a subset $S\subseteq V$ with minimum cardinality such that for every vertex 
$v\in V\backslash S$, there is a vertex $u\in S$ such that there is a path between them of length at most $k$. The minimum total $k$-distance dominating set is defined similarly.

Also, a subset of vertices such that each edge of the graph is incident to at least one vertex of the subset is a vertex cover. Minimum vertex cover (MVC) is a vertex cover having the smallest possible number of vertices for a given graph. The size of MVC is shown by  $\beta(G)$.

An interesting problem is to compute the minimum dominating set and the minimum vertex cover in distributed model.
In a distributed model the network is abstracted as a simple $n$-node undirected graph $G = (V,E)$. There is one processor on each graph node $v \in V$, with a unique $\Theta(log n)$-bit identifier $ID(v)$, who initially knows only its neighbors in $G$. Communication happens in synchronous rounds. Per round, each node can send one, possibly different, $O(\log n)$-bit message to each of its neighbors. At the end, each node should know its own part of the output. For instance, when computing the dominating set, each node knows whether it is in the dominating set or has a neighbor in the dominating set \cite{ghaf}.    

Computing minimum dominating set has many real world applications. 
Nowadays online social networks are growing exponentially and they have important effect on our daily life. 
To influence the network participants a key feature in a social network is the ability to communicate quickly within the network. For example, in an emergency situation, we may need to be able to reach to all network nodes, but only a small number of individuals in the network can be contacted directly due to the time or other constraints. However, if all nodes from the network are connected to at least one such individual who can be contacted directly (or is one of those individuals) then the emergency message can be quickly sent to all network participants. In this scenario the goal is to choose the minimum number of such nodes. A challenge is that each node knows its rule instantly.

Also in wireless networks consider the scenario where in order to maximize survivability, the battery power can be conserved by having the minimum possible active sensors, especially for sensors with wide overlapping fields of view. So, we need to find a minimum subset of sensors that need to remain active in order to provide a desirable level of coverage. This scenario is presented in \cite{sol}. 
As another scenario, consider a group of mobile robots each with a wireless access point. The goal of the robots is to maximally cover an area with the wireless network. As the robots are traveling between waypoints, though, it is highly likely that there will be a large amount of overlap in the coverage. Therefore, in order to save power, the robots might want to choose a maximum subset of robots that can lower their transmit power while still retaining coverage\cite{sol}. Now suppose that the robots are chosen but suddenly some of them need to be repaired, so the solution should be changed accordingly.
The challenge in each of these scenarios is for the agents to collectively find the solution without relying on centralization of computation. Centralization is infeasible either due to lack of resources (i.e., no single agent has powerful enough hardware to solve the global problem) or due to lack of time (i.e., centralizing the problem will take at least a linear number of messaging rounds).
Another challenge is that how to solve the problem when some of the inputs are changed or extra constraints are added, for example suppose that some areas should be covered by a specific set of robots or nodes in the network. 

The Art Gallery Problem (AGP), a well known problem in computational geometry community, is another problem that is related to MDS problem. 
There are practical problems that turn out to be related to AGP. Some of these are straightforward, such as guarding a shop with security cameras, or illuminating an environment with few lights. Also the AGP arises in multiagent systems. For example, many robotics, sensor network, wireless networking, and surveillance problems can be mapped to variants of the art gallery problem \cite{sol}. The nature of these problems leads us to apply multiagent paradigm, each guard is considered as an agent.
So, a new research area is considering the AGP in multiagent paradigm and in distributed model.
The AGP is equivalent to the Coverage Problem
in the context of wireless sensor networks, wireless ad-hoc networks, and
wireless sensor ad-hoc networks \cite{meg}.

There are also many other problems in networking that can be modeled as the minimum dominating set problem. For example,
in \cite{wan}, the problem of node placement
for ensuring complete coverage in a long belt with minimum number of nodes scenario is studied. Each node is assumed to be able to cover a disk area centered at itself with a fixed radius.
In \cite{cha} grid coverage for surveillance and target location in distributed sensor networks is studied.

In social networks the minimum $k$-distance dominating set can be considered as social recommenders. 
The close nodes influence each other and they have the same preferences in a network.
Suppose that we want to give recommendation on a special product (e.g. which movie to watch) to each node of network but we can't reach all of them because of the time constraint and advertising cost. We may choose minimum number of nodes such that they dominate all other nodes within distance $k$ from them. We give a recommendation to each of the selected nodes and then they spread it in the network. This is equal to solving $k$-dominating set problem. For more on social recommendation see \cite{avn}. 

\subsection{Recent results and related works}
\subsection*{Sequential model}
Finding a minimum dominating set is NP-complete \cite{karp}, even for planar graphs of maximum degree $3$ \cite{gar},
and cannot be approximated for general graphs with a constant ratio under the assumption $P\neq NP$ \cite{raz}.
An $O(log n)$-approximation factor can be found by using a simple greedy algorithm. Moreover, negative results have been proved for the approximation of MDS even when limited to the power law graphs \cite{gas}.
 A number of works have been done on exact algorithms for MDS, which mainly focus on improving the upper bound of running time. State of the art exact algorithms for MDS are based on the branch and reduce paradigm and can achieve a run time of $O({1.4969}^n)$ \cite{van}. Fixed parameterized algorithms have allowed to obtain better complexity results \cite{kar}. The main focus of such algorithms is on theoretical aspects.

In practice, these theoretical algorithms are not applicable specially in massive networks because of time and space constraints. So we need to use heuristic algorithms  to obtain solutions. See \cite{san} for a  comparison among several greedy heuristics for MDS. 

In sequential model heuristic search methods such as genetic algorithm \cite{her} and ant colony optimization \cite{pot11,pot13} have been developed to solve MDS. Also Hyper metaheuristic algorithms combine different heuristic search algorithms and preprocessing techniques to obtain better performance \cite{pot13,sach,blum,gen,abu}. 
These algorithms were tested on standard benchmarks with up to thousand vertices. The configuration checking (CC) strategy \cite{cai11} has been applied to MDS and led to two local search algorithms. Wang et al. proposed the CC2FS algorithm for both unweighted and weighted MDS \cite{w}, and obtained better solutions than ACO-PP-LS  \cite{pot13} on standard benchmarks. Afterwards, another CC-based local search named FastMWDS was proposed, which significantly improved CC2FS on weighted massive graphs \cite{wang18}. Chalupa proposed an order-based randomized local search named RLSo \cite{chal}, and achieved better results than ACO-LS and ACO-PP-LS \cite{pot11,pot13} on standard benchmarks of unit disk graphs as well as some massive graphs. Fan et. al. designed a local search algorithm named ScBppw \cite{fan}, based on two ideas including score checking and probabilistic random walk. Recently an efficient local search algorithm for MDS is proposed in \cite{cai20}. The algorithm named FastDS is evaluated on some standard benchmarks. FastDS obtains the best performance for almost all benchmarks, and obtains better solutions than previous algorithms on massive graphs in their experiments.
A recent study for the $k$-dominating set problem can be found in \cite{min}. They proposed a heuristic algorithm that can handle real-world instances with up to $17$ million vertices and $33$ million edges. They stated that this is the first time such large graphs are solved for the minimum $k$-dominating set problem. They compared their proposed algorithm with the other best known algorithms for this problem.

\subsection*{Distributed model}

The centralized algorithms for the MDS and the MTDS problems have been studied well in the literature. However, there is little known about distributed algorithms for these problems. 
Most of the distributed algorithms proposed to solve the dominating set problem lack giving bounds on both runtime and solution quality. Most of the time the emphasis in the wireless networking community and social networks is on algorithms with a constant number of communication rounds. For example, Ruan et.al. in \cite{ruan} proposed a one-step greedy approximation algorithm for the minimum connected dominating set problem (MCDS), with an approximation factor  that is a function of $\Delta(G)$,  where $\Delta(G)$ is the maximum degree of the graph $G$. Kuhn and Wattenhofer \cite{kuhn} proposed a more general result, their approximation factor is variable and a function of the number of communication rounds. However, this algorithm also depends on $\Delta$. Huang et.al. in \cite{chu}, by increasing the length of messages in each communication round, gave a $12$-approximation algorithm for MCDS problem. For more theoretical results on distributed algorithms for MDS problem see \cite{akh}.

 In \cite{hike}, it has been shown that for any $\epsilon>0$ there is no deterministic local algorithm that finds a $(7-\epsilon)$-
approximation of a minimum dominating set for planar graphs. However, there exist an algorithm with approximation factor of $52$ for computing a MDS in planar graphs \cite{and,chris} in local model and an algorithm with approximation factor of $636$ for anonymous networks \cite{and,waw}. In \cite{ali1}, they improved the approximation factor in anonymous networks to $18$ in planar graphs without $4$-cycles. For more information on local algorithms see \cite{suo}.
Then in \cite{ali2}  it has been proved that the approximation factor of \cite{ali1} for triangle-free planar graphs is 32 and 16 for MDS and MTDS.
They have also presented a modified version of the algorithm presented in \cite{ali1} and implemented their algorithm on real data sets.

Sultanik  et al. \cite{sol} introduced a distributed algorithm for the art gallery and dominating set problem that is guaranteed to run in a number communication rounds on the order of the diameter of the visibility graph. They show through empirical analysis that the algorithm will produce solutions within a constant factor of optimal with high probability.
The version of AGP that they studied is equivalent to computing MDS of the visibility graph of polygons.

\subsection{Our results}

In this paper, first we present our theoretical results about computing MDS and MTDS of graphs. We give upper bounds for MDS and MTDS and fast distributed randomized algorithms to achieve these bounds.
This upper bound is similar to Caro-Wei bound for maximum independent set of graphs (see \cite{caro} and \cite{wei}).
 Next we propose our algorithms for computing the minimum dominating set of graphs using our theoretical results. 
 


In the distributed model the first algorithm runs in constant number of rounds and the communication rounds of the second one depends on the distributed algorithms that are used to find a minimum vertex cover.  
For example in \cite{val} a $3$-approximation for $\beta(G)$ 
is given that runs in $2\Delta+1$ rounds. In \cite{matt} a $2$-approximation algorithm for $\beta(G)$ is given that runs in $(\Delta+1)^2$ rounds.

Our algorithms can be run in dynamic model where the nodes are added or deleted constantly as well. We can handle the case where each node should be dominated by a special set of nodes.
Also we show how to to extend the algorithms to solve $k$-distance dominating set problem, and set cover problem.

\section{Theoretical Results}
In this section, we present our main idea for computing minimum dominating set (MDS) and minimum total dominating set (MTDS) of a given graph $G=(V,E)$. Then we present an upper bound for MDS and MTDS and we give a distributed randomized algorithm for computing this upper bound.
\subsection{Main idea}
First we present our idea for computing MTDS.
Here, we assume all considered graphs have no isolated vertex.
For a given graph $G$, we construct a graph $G'$ with the same set of vertices as in $G$. For each vertex $v$,  we choose two of its neighbors arbitrarily and add an edge between them in $G'$, if $v$ is of degree one, we add a loop edge on its neighbor (See Fig \ref{G1} and \ref{G2}). We call this edge the corresponding edge of $v$ in $G'$ and denote it by $e_v$. Note that if the graph $G$ has a cycle of length $4$, with vertices $a,b,c,d$ then the edge $bd$ can be the corresponding edge of both $a$ and $c$ in $G'$. So, $|E(G')|\leq n$ ($|E(G')|$ denotes the size of set $E(G')$ and $n$ is the size of the set $V(G)$).
 Obviously the construction of $G'$ can be done in one round in the distributed model.
 Let $\alpha(G)$ and $\beta(G)$ be the size of maximum independent set and the size of minimum vertex cover of $G$, respectively. 
\begin{lemma}
\label{1}
$\gamma_t(G)\leq n-\alpha(G')=\beta(G').$
\end{lemma}
\begin{proof}
Suppose that $D$ is a maximum independent set of $G'$, so $|D|=\alpha(G')$.
Now we show that $V\backslash D$ is a total dominating set for $G$. For each vertex $v$, we choose two of its neighbors, for example $u$ and $w$, and add an edge between them in $G'$. Since there is an edge between $u$ and $w$, so at most one of them can be in $D$ which means at least one of them is in $V\backslash D$.
The same argument applies when an edge is a loop. Thus for each vertex $v$ at least one of its neighbors in $G$ is in $V\backslash D$, so $V \backslash D$ is a total dominating set for $G$. The size of $|V\backslash D|$ equals $n-\alpha(G')$  and we have 
$\gamma_t(G)\leq n-\alpha(G')=\beta(G')$.
\end{proof}
Note that $G'$ is not unique, so there is a set $A$ of graphs $G'_i$'s such that they can be constructed as we explained earlier.
Now we present our main theorem.
\begin{theorem}
\label{main}
Let $G'_{min}\in A$ be such that $$n-\alpha(G'_{min})=min_{G'_i\in A}{n-\alpha(G'_i)},$$ then $$\gamma_t(G)=n-\alpha(G'_{min})=\beta(G'_{min}).$$
\end{theorem}
\begin{proof}
By Lemma \ref{1} we have $\gamma_t(G)\leq \beta(G'_{min})$. To show that $\beta(G'_{min})\leq \gamma_t(G)$,
 it is enough to construct a graph $G'\in A$ such that $\beta(G')=\gamma_t(G)$.
Let $S$ be a total dominating set of vertices of cardinality $\gamma_t(G)$.  
For each vertex  $v\in V$, there is at least one vertex $u\in S$ such that $u$ and $v$ are adjacent. If $d(v)=1$ then, we put a loop on its neighbor, $u$, otherwise $v$ has at least another neighbor, for example $w$.
We put an edge between $u$ and $w$. Now we have our graph $G'\in A$.
Since every edge of $G'$ has at least one of its endpoints in $S$, hence $S$ is a vertex cover for $G'$. On the other hand, any vertex cover for $G'$ is a total dominating set for $G$, since for any vertex $v$ of $G$ there is an edge of $G'$ whose endpoints are adjacent to $v$, hence it is dominated by a vertex cover of $G'$. Therefore $S$ is a minimum dominating set for $G'$ and we have $\beta(G')=\gamma_t(G)$  (See Fig \ref{G3}).

\end{proof}

\begin{figure}[H]
  \centering
  \caption{Graph $G$ with $\gamma_t(G)=3$ and the constructed graphs from $G$.}
  
\begin{subfigure}{1.6in}
\includegraphics[width=\textwidth]{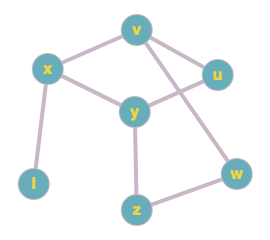}
\caption{Graph $G$ with a MDTS of size 3. For example $\{v,x,y\}$ is a MTDS for $G$.}
\label{G1}
\end{subfigure}
  
\begin{subfigure}{1.6in}
\includegraphics[width=\textwidth]{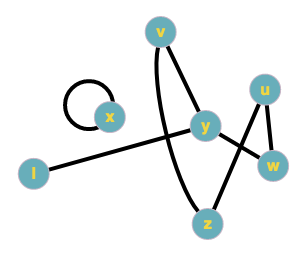}
\caption{Graph $G'$ which is constructed from $G$. Here, $e_v=uw$, $e_u= vy$, $e_y= uz$, $e_w= vz$, $e_z= wy$, $e_x= vy$ and $e_i= xx$, where by $xx$ we mean a loop on $x$ and $\beta(G')=4$. For example $\{x,y,z,w\}$ is a minimum vertex cover for $G'$ and a dominating set for $G$.}
\label{G2}
\end{subfigure}
  
\begin{subfigure}{1.6in}
\includegraphics[width=\textwidth]{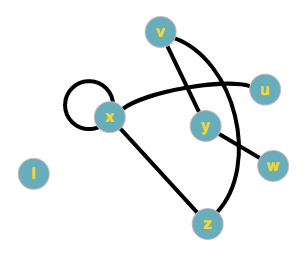}
\caption{Here the Graph $G'_{min}$ is also constructed from $G$. Here, $e_v= ux$, $e_u= vy$, $e_x= vy$, $e_y= xz$, $e_w= vz$, $e_z=wy$ and $e_i=xx$ and $\beta(G'_{min})=3$. For example $\{x,v,y\}$ is a minimum vertex cover for $G'_{min}$ and a minimum total dominating set for $G$.}
\label{G3}
\end{subfigure}
\label{G}
\end{figure}

Now we give a similar argument for computing MDS. We construct a graph $G''$ from $G$ as follows. The vertex set of $G''$ is the same as $G$. For each vertex $v\in V$ we choose two vertices from $N(v)\bigcup v$  and add an edge between those two selected vertices. This way we have a graph $G''$.
For MDS we do not have the assumption that there is no isolated vertices. So if $v$ is isolated we simply put a loop on $v$.
\begin{lemma}
\label{11}
$\gamma(G)\leq n-\alpha(G'')=\beta(G'').$
\end{lemma}
Since $G''$ is not unique, so there is a set $A$ of graphs $G''_i$'s such that they can be constructed as we explained earlier.
We have the following theorem.
\begin{theorem}
\label{main1}
Let $G''_{min}\in A$ be such that  $$n-\alpha(G''_{min})=min_{G''_i\in A}{n-\alpha(G''_i)},$$ then $$\gamma(G)=n-\alpha(G''_{min})=\beta(G''_{min}).$$

\end{theorem}
We omit the proofs because they are similar to the proofs of Lemma \ref{1} and Theorem \ref{main}. As we said earlier $G'$ and $G''$ are not unique. So, our aim is to give algorithms for constructing $G'$ and $G''$ such that $\beta(G')$ and $\beta(G'')$ be a good approximation of $\beta(G'_{min})$ and $\beta(G''_{min})$.

As a result of Lemma \ref{1} and Lemma \ref{11} in the following subsection we present an upper bound for MDS and MTDS.

\subsection{An upper bound for MDS and MTDS and a distributed algorithm for computing that bound}

To compute an upper bound for MTDS of a graph $G$, we construct a graph $G'$ as before. By Lemma \ref{1}, $\gamma_t(G)\leq n-\alpha(G')$. 
Caro \cite{caro} and Wei \cite{wei} showed  that in a given graph $G$, $\sum^n_{i=1}\frac{1}{1+d_i}\leq \alpha(G)$, where $d_i$ is the degree of vertex $i$.
So, we have the following theorem.
\begin{theorem}
By Lemma \ref{1} and Lemma \ref{11} if we construct the graphs $G'$ and $G''$ arbitrarily then,
 $$\gamma_t(G)\leq n-\alpha(G')\leq n-\sum^n_{i=1}\frac{1}{1+d_i'}.$$ Where $d_i'$ is the degree of $v_i$ in $G'$. 
 And 
  $$\gamma(G)\leq n-\alpha(G'')\leq n-\sum^n_{i=1}\frac{1}{1+d_i^{''}}.$$ Where $d_i^{''}$ is the degree of $v_i$ in $G''$.
\end{theorem}

 These bounds for MTDS and MDS is similar to Caro and Wei bound for independent set of graphs. 

\subsection*{A randomized distributed algorithm for computing the given upper bound for MTDS and MDS}

An independent set of expected size $\sum^n_{i=1}\frac{1}{1+d_i}$ for a graph $G$ can be found by a simple linear time randomized algorithm
that follows from an analysis of the Caro-Wei bound by Alon and Spencer in \cite{alon}. This algorithm works as follows. Every node $v$ chooses a random real value between $0$ and $1$ and adds itself to the independent set $I$ if none of its neighbors
have chosen a larger real value than $v$. Then, the probability that a node $v$ is
added to the independent set is $\frac{1}{1+d_i}$, hence by linearity of expectation, $E(I)=\sum^n_{i=1}\frac{1}{1+d_i}$.

So, our distributed algorithm to compute an upper bound for MTDS is as follows. We construct a graph $G'$ arbitrary as explained and then compute an independent set of expected size of $\sum^n_{i=1}\frac{1}{1+d'_i}$. We choose those vertices that are not in the independent set of $G'$. Which means that we compute a total dominating set of expected size of  $n-\sum^n_{i=1}\frac{1}{1+d'_i}$ in constant number of rounds. The same argument applies for the MDS. We construct a graph $G''$  and then we compute a vertex cover of expected size of $n-\sum^n_{i=1}\frac{1}{1+d^{''}_i}$.

\section{Algorithms}

In this section we present two distributed algorithms for computing a dominating set for a given graph.

\subsection{First Algorithm}

The first algorithm is the same as algorithm presented in \cite{ali2} with a small modification. In \cite{ali2}, they compute a total dominating set and since a total dominating set is also a dominating set in graphs with no isolated vertex, they consider this total dominating set as a dominating set. But in our modified version we compute a dominating set. As we said earlier the size of MTDS can be twice of the size of MDS so in practice we expect that this algorithm performs better than the algorithm in \cite{ali2}.

\begin{algorithm}
\begin{algorithmic}[1]
\label{alg1}
\caption{First distributed algorithm for computing a dominating set in a graph with given integer $m\ge 0$.}
\label{alg1}
\STATE In the first round, each node $v_i$ chooses a random number $0<r_i<1$ and computes its weight $w_i=d_i+r_i$ and sends $w_i$ to its adjacent neighbors.
\STATE In the second round, each node $v$ marks a vertex $v_i\in N(v)\bigcup v$, whose weight $w_i$  is maximum among all the other neighbors of $v$.
\FOR {$m$ rounds} 
  \STATE Let $x_i$ be the number of times that a vertex is marked by its neighbor vertices, let $w_i=x_i+r_i$
  \STATE Unmark the marked vertices.
  \STATE Each vertex $v$ marks the vertex with $v_i\in N(v)\bigcup v$ maximum $w_i$.
    \ENDFOR
\STATE The marked vertices are considered as the vertices in our dominating set for $G$.
\end{algorithmic}
\end{algorithm}

Obviously the set of marked vertices is a dominating set since each vertex marks itself or one of its neighbors.  
In the first round, $r_i$'s are generated and added to $d_i$'s. And in the next round each vertex mark a vertex with maximum $w_i$. In the next $m$ rounds each vertex marks a vertex based on $x_i$'s.
So, in a distributed network for a constant number $m$, this algorithm runs in constant number of rounds.
Note that this algorithm is the same as \cite{ali2}, except that in line 2 and line 6 of Algorithm \ref{alg1} each vertex marks a vertex in $v\bigcup N(v)$ but 
 in \cite{ali2}, each vertex $v$ marks a vertex in $N(v)$. That is why their algorithm gives a total dominating set but Our algorithm gives a dominating set.
 
 \subsection{Second Algorithm}

In the second algorithm our aim is to improve the results of Algorithm \ref{alg1} by using Theorem \ref{main1}.
First we run Algorithm \ref{alg1}. In this algorithm at the last step for each $v_i$ we know the value of $x_i$, i.e. the number of times that vertex $v_i$ is selected by its neighbors or itself.
We construct a graph $G''_1$ from $G$ as follows. The vertex set of $G''_1$ is the same as $G$. For each vertex $v_i$ in $G$ we choose two vertices  $u_i,y_i\in N(v_i)\bigcup v_i$ with maximum value of  $x_j+r_j$'s such that $x_j>0$. Then we add an edge between $u_i$ and $y_i$ in $G''_1$ and if there is only one $x_j>0$ we add a loop on $v_j$.
\begin{algorithm}[H]
\caption{Second distributed algorithm for computing dominating set in a given graph. }
\label{alg2}

\begin{algorithmic}[1] 
\STATE Run Algorithm \ref{alg1} for $m$ rounds.
\STATE Let $G''_1$ be a graph with the same vertex set as $G$.
\STATE For each vertex $v_i$ choose two vertices  $u_i,y_i\in N(v_i)\bigcup v_i$ with maximum values of  $x_j+r_j$'s such that $x_j>0$. Add an edge between $u_i$ and $y_i$ in $G''_1$ and if there is only one $x_j>0$ we add a loop on $v_j$.
\STATE Compute a vertex cover for $G''_1$ which is a dominating set for $G$.
\end{algorithmic}
\end{algorithm}

By Lemma \ref{1}, if we compute a vertex cover for $G''_1$ then the vertices in the vertex cover of $G''_1$ form a dominating set for $G$.
Obviously $G''_1$ can be constructed in constant number of rounds in distributed model.
There are well known algorithms for computing the minimum vertex cover of a graph in distributed model which according to our running time and space constraints we can use one of them.

In this section we explained two algorithms for computing a dominating set for graphs. We are not able to compute the approximation factor of these algorithms theoretically in general. Instead we implement the algorithms on real data sets and compare the results with state of the art algorithms.

\section{Experiments}

\subsection*{Data description}

In the following we present a brief description of the benchmarks from \cite{cai20}.
 
T1\footnote{http://mail.ipb.ac.rs/~rakaj/home/BenchmarkMWDSP.htm}: This data set consists of $520$ instances where each instance has two different weight functions. As in \cite{cai20} we select these original graphs where the weight of each vertex is set to $1$. There are $52$ families, each of which contains $10$ instances with the same size. 
 The instances have 50 to 1000 nodes with different number of randomly created edges but always making the graph connected (for more details see \cite{rom}).

BHOSLIB\footnote{http://networkrepository.com/bhoslib.php}: This benchmark are generated based on the RB model near the phase transition. It is known as a popular benchmark for graph theoretic problems. The order of average number of vertices in this benchmark is about $750$ and average number of edges is about $10^5$ (for more details see \cite{nr}).

SNAP\footnote{http://snap.stanford.edu/data}: This benchmark is from Stanford Large Network Dataset Collection. It is a collection of real world graphs from $10^4$ vertices to $10^7$ vertices (for more details see \cite{snap}).

DIMACS10\footnote{http://networkrepository.com/dimacs10.php}: This benchmark is from the 10th DIMACS implementation challenge, which aims to provide large challenging instances for graph theoretic problems (for more details see \cite{nr}).

Network Repository\footnote{http://networkrepository.com/}: The Network Data Repository includes massive graphs from various areas. Many of the graphs have $100$ thousands or millions of vertices. This benchmark has been widely used for graph theoretic problems including vertex cover, clique, coloring, and dominating set problems.
As in \cite{cai20} for SNAP benchmarks we consider the graphs with at list $30000$ vertices and for Repository benchmark we choose the graphs with at least $10^5$ vertices (for more details see \cite{nr}).

\subsection*{Experimental results and implementation}

The most related work to ours is in \cite{ali2}, where their local distributed algorithm computes a total dominating set for graphs and since a total dominating set is also a dominating set so they implemented and ran their algorithm on some real data sets and compared their results with a recent centralized algorithm for minimum dominating set problem in \cite{cai20}.

In Table \ref{t1}, Table \ref{bho}, Table \ref{dim} and Table \ref{nr} we present our results and compare the results with \cite{ali2}.
The first column  is the output of Algorithm \ref{alg2} with two modifications. The first modification is that instead of line 1, we run algorithm presented in \cite{ali2} for $m=5$ iterations. And the second modification is that in line 4, instead of choosing $u_i$ and $y_i$ from $N(v_i)\bigcup v_i$, we choose them from $N(v_i)$. In this case the achieved dominating set  is also a total dominating set.  We call this modified 1 of Algorithm \ref{alg2} (Mod1).
In the second column we run Algorithm \ref{alg2} with a modification in line $1$ as follows. Instead of algorithm \ref{alg1} we run the algorithm in \cite{ali2} for $m=5$  iterations. We call this Modified 2 of Algorithm \ref{alg2} (Mod2).
In the third and forth columns the results of Algorithm \ref{alg2} and Algorithm \ref{alg1} for $m=5$ are presented. 
In sixth column the results of implementation of the algorithm in \cite{ali2} is presented.
And in the last column we present the results of  algorithm presented in \cite{cai20}.  The empty cells were not computed in \cite{cai20}.

\begin{table}[h]
\begin{footnotesize}
\caption{Experimental results for T1 benchmark.}
\begin{center}
\begin{tabular}{| lllll  ll |}
 \hline
 \textbf{Instance}&\textbf{Mod1}&\textbf{Mod2}&\textbf{Alg \ref{alg2}}&\textbf{Alg \ref{alg1}}&\textbf{\cite{ali2}}&\textbf{\cite{cai20}}\\ \hline
 
V100E100		&42	&42	&42&50	&50&34\\ \hline
V100E1000	&10&10&	10&	12&	12&8\\ \hline
V100E2000&	6&	6&	6&	6&	6&5\\ \hline
V100E250	&	28&	28&	28&	29&	29&20\\ \hline
V100E500	&	16&	16&	16&	18&	18&13\\ \hline
V100E750	&	12&	12	&12&12&	12&9\\ \hline
V150E1000&	21&	21	&21&22	&22&15\\ \hline
V150E150	&	63&	63&	63&	73&	73&50\\ \hline
V150E2000&	12&	12&	12&	13&	13&9\\ \hline
V150E250	&	46&	46&	46&	51&	51&39\\ \hline
V150E3000&	11&	11&	11&	11&	11&7\\ \hline
V150E500	&	31&	31	&31&36&	36&25\\ \hline
V150E750	&	25&	25	&25&26&26&18\\ \hline
V200E1000&	32&	32&	32&	34&	34&24\\ \hline
V200E2000&	22&	22&	22&	22&	22&15\\ \hline
V200E250	&	80&	80&	80&	87&	87&61\\ \hline
V200E3000&	13&	13&	13&	14&	14&11\\ \hline
V200E500	&	54&	54&	54&	55&	55&37\\ \hline
V200E750	&	41&	41&	41&	44&	44&30\\ \hline
V250E1000&	49&	49&	49&	52	&52&36\\ \hline
V250E2000&	30&	30	&30&31&31&22\\ \hline
V250E250	&	106&	106&	106&	122&	122&83\\ \hline
V250E3000&	24&	24&	24&	25	&25&16\\ \hline
V250E500	&	77&	77&	77&	87&	87&58\\ \hline
V250E5000&	16&	16&	16&	17	&17&11\\ \hline
V250E750	&	62&	62&	62&	64&	64&44\\ \hline
V300E1000&	62&	62&	62&	70&	70&49\\ \hline
V300E2000&	40&	40&	40&	43	&43&29\\ \hline
V300E300	&	128&	128&	128&	141	&141&100\\ \hline
V300E3000&	30&	30&	30&	31&	31&22\\ \hline
V300E500&	107&	107&	107&	115&	115&78\\ \hline
V300E5000&	22&	22&	22&	23&	23&15\\ \hline
V300E750	&	86&	86&	86&	89&	89&60\\ \hline
V500E1000&	150&	150&	150&	165&	165&115\\ \hline
V500E10000&	35&	35&	35&	35&	35&22\\ \hline
V500E2000&	104&	104&	104	&114&114&71\\ \hline
V500E500	&	214&	214&	214	&241	&241&167\\ \hline
V500E5000&	58&	58&	58&	59&	59&37\\ \hline
V800E1000&	326&	326&	326&	355&	355&267\\ \hline
V800E10000&	80&	80&	80&	85&	85&50\\ \hline
V800E2000&	222&	222&	222&	239&	239&158\\ \hline
V800E5000&	121&	121&	121&	129&	129&83\\ \hline
V1000E1000&	417&	417&	417&	476&476&334\\ \hline
V1000E10000&	110&	110	&110	&118&118&74\\ \hline
V1000E15000&	80&	80&	80&	85&	85&55\\ \hline
V1000E20000&	73&	73&	73&	76&	76&45\\ \hline
V1000E5000&	184&	184&	184&	194&	194&121\\ \hline

 \end{tabular}
\end{center}
\label{t1}
\end{footnotesize}
\end{table}

\begin{table}[h]
\begin{footnotesize}
\caption{Experimental results for BHOSLIB benchmark.}

\begin{tabular}{| lllll ll |}
 \hline
 \textbf{Instance}&\textbf{Mod1}&\textbf{Mod2}&\textbf{Alg \ref{alg2}}&\textbf{Alg \ref{alg1}}&\textbf{\cite{ali2}}&\textbf{\cite{cai20}}\\ \hline

frb40-19-1&2& 2& 2&3&3&14\\ \hline
frb40-19-2 &3&3& 4&4&3&14\\ \hline
frb40-19-3 &4&4&4& 4&4&14\\ \hline
frb40-19-4 &3&3&3&4&4&14\\ \hline
frb40-19-5 &3& 3& 3& 3&3&14\\ \hline
frb45-21-1 &3&3&3&4&4&16\\ \hline
frb45-21-2 &4&5&3&4&5&16\\ \hline
frb45-21-3 &3&3&3&3&3&16\\ \hline
frb45-21-4 &4&4&4&4&4&16\\ \hline
frb45-21-5 &3&3& 3&3&3&16\\ \hline
frb50-23-1 &3&3&3&4&4&18\\ \hline
frb50-23-2 &4&4&4&4&4&18\\ \hline
frb50-23-3 &3&3& 3&3&4&18\\ \hline

\end{tabular}

\label{bho}
\end{footnotesize}
\end{table}

In BHOSLIB the achieved results are surprisingly better than \cite{cai20}.  In dense graphs or the graphs with large maximum degree, close to n (number of vertices), the adjacent vertices to the maximum degree choose it, so the number of marked vertices is small and close to the exact solution. For example in instance frb40-19-1 of BHSLIB benchmark, the number of vertices is about 760, the maximum degree is 703, the minimum degree is 581 and the average degree is 650.

\begin{table*}[h]
\begin{footnotesize}
\caption{Experimental results for Network snap and DIMACS10 benchmark.}
\begin{center}
\begin{tabular}{| lllll ll |}
\hline
 \textbf{Instance}&\textbf{Mod1}&\textbf{Mod2}&\textbf{Alg \ref{alg2}}&\textbf{Alg \ref{alg1}}&\textbf{\cite{ali2}}&\textbf{\cite{cai20}}\\ \hline

Amazon0302(V262K E1.2M)&	46602&	43965&	42095&	45742&	49903&35593\\ \hline
Amazon0312(V400K E3.2M)&	56034&	53640&	52707&	57068&	59723&45490\\ \hline
Amazon0505(V410K E3.3M)&	58088&	55717&	54687&	59241&	61905&47310\\ \hline
Amazon0601(V403K E3.3M)&	52132&	50298&	49464&	53432&	55644&42289\\ \hline
email-EuAll(V265K E420K)&	33852&	31468&	18185&	18219&	33864&18181\\ \hline
p2p-Gnutella24(V26K E65K)&	5557	 &5515&	5476 &	5655&	  5718&5418\\ \hline
p2p-Gnutella25(V22K E54K)&	4645	&4610	&4594&	4756	& 4807&4519\\ \hline
p2p-Gnutella30(V36K E88K)&	7336&	 7281&	7263 &	7449 &	7524&7169\\ \hline
p2p-Gnutella31(V62K E147K)&	12793&	12703&	12676&	12980&	13115&12582\\ \hline
soc-sign-Slashdot081106(V77K E516K)&	14975&	14420&	14390&	14865&	15209&14312\\ \hline
soc-sign-Slashdot090216(V81K E545K)&	16118&	15484&	15446&	16010&	16418&15305\\ \hline
soc-sign-Slashdot090221(V82K E549K)&	16154&	15517&	15490&	16021&	16436&-\\ \hline
soc-Epinions1(V75K E508K)&	16557&	15840&	15789&	16255&	16760&15734\\ \hline
web-BerkStan(V685K E7.6M)&	37711&	35039&	31784&	33980&	39938&28432\\ \hline
web-Stanford(V281K E2.3M)&	18350&	16887&	15032&	16176	&19643&13199\\ \hline
wiki-Talk(V2.3M E5M)&	40135&	39324&	36969&	37219&	40191&36960\\ \hline
wiki-Vote(V7K E103K)&	1153&	 1143&	1121 &	1150 &	1177&1116\\ \hline
cit-HepPh(V34K E421K)&	3812	& 3701&	3624 &	3905&	 4074&3078\\ \hline
cit-HepTh(V27K E352K)&	3764	& 3553&	3386 &	3684&	 4025&2936\\ \hline
rgg-n-2-17-s0&	21605&	20495&	19282&	21430&	23523&43412\\ \hline
rgg-n-2-19-s0&	80038&	76081&	71875&	79557&	86742&844423\\ \hline
rgg-n-2-20-s0&	153393&	146039&	138412&	153129&	166522&84708\\ \hline
rgg-n-2-21-s0&	295851&	281372&	267307&	295447&	320642&162266\\ \hline
rgg-n-2-22-s0&	571868&	545511&	518014&	572266&	619551&312350\\ \hline
rgg-n-2-23-s0&	1107851&	1057773&	1006686&	1110615&	1199233&605278\\ \hline
coAuthorsCiteseer&	37005&	34508&	34139&	36381&	38310&22011\\ \hline
co-papers-citeseer&	34647&	32114&	31330&	34874&	37057&26082\\ \hline
kron-g500-logn16&	14120&	14118&	14117&	14171&	14174&14100\\ \hline
co-papers-dblp	&48638&	45467&	44805&	49821	&52187&43978\\ \hline

\end{tabular}
\end{center}
\label{dim}
\end{footnotesize}
\end{table*}

\begin{table*}[h]
\begin{footnotesize}
\caption{Experimental results for Network repository benchmark.}
\begin{center}
\begin{tabular}{| lllll ll |}
\hline
 \textbf{Instance}&\textbf{Mod1}&\textbf{Mod2}&\textbf{Alg \ref{alg2}}&\textbf{Alg \ref{alg1}}&\textbf{\cite{ali2}}&\textbf{\cite{cai20}}\\ \hline

soc-youtube(V496 E2M)&	99669&	92020&	91192&	96212&	102355&89732\\ \hline
soc-flickr(V514K E3M)&	104571&	99237&	98832&	102194&	106337&98062\\ \hline
ca-coauthors-dblp(V540K E15M)&	48647&	45533&	44833&	49841&	52180&35597\\ \hline
ca-dblp-2012(V317K E1M)&	50246&	47497&	47067&	49669&	51790&46138\\ \hline
ca-hollywood-2009(V1.1 E56.3)&	58060&	57072&	56972&	61096&	61626&48740\\ \hline
inf-roadNet-CA(V2M E3M)&	834653&	785263&	718224&	790165&	911273&586513\\ \hline
inf-roadNet-PA(V1M E2M)	&464398&	436853&	400628&	440939&	507130&326934\\ \hline
rt-retweet-crawl(V1M E2M)&	82927&	76039&	75825&	76916&	83368&75740\\ \hline
sc-ldoor(V952k E21M)&	77595&	75189&	70543&	73017&	79629&62411\\ \hline
sc-pwtk(V218K E6M)&	8783&	8228	&7321&	8077	&9444&4200\\ \hline
sc-shipsec1(V140K E2M)&	13908&	13638&	13361&	14405&	14926&7662\\ \hline
sc-shipsec5(V179K E2M)&	20512&	20179&	19940&	21689&	22184&10300\\ \hline
soc-FourSquare(V639K E3M)&	61324&	61324&	61324&	62053&	62053&60979\\ \hline
soc-buzznet(V101K E3M)&	138&	138	&138&	150&	150&127\\ \hline
soc-delicious(V536K E1M)&	57795&	56192&	56067&	57131&	58491&55722\\ \hline
soc-digg(V771K E6M)&	70185&	67240&	66896&	69234&	71889&66155\\ \hline
soc-flixster(V3M E8M)&	91528&	91044&	91035&	91312&	91605&91019\\ \hline
soc-lastfm(V1M E5M)&	67445	&67270&	67258&	67466&	67621&67226\\ \hline
soc-livejournal(V4M E28M)&	855807&	826813&	822403&	868615	&891958&793887\\ \hline
soc-orkut(V3M E106M)&	141426&	141267&	141208&	151742&	151881&110547\\ \hline
soc-pokec(V2M E22M)&	234696&	231289&	230622	&245806&	248740&207308\\ \hline
soc-youtube-snap(V1M E3M)&	231538&	215321&	214338&	222480&	234965&213122\\ \hline
socfb-FSU53(V28K E1M)&	2388&	2379	&2369&	2575	&2589&-\\ \hline
socfb-Indiana69(V30K E1M)&	2301&	2289	&2278&	2435	&2450&-\\ \hline
socfb-MSU24(V32K E1M)	&2837&	2806	&2797&	2996	&3020&-\\ \hline
socfb-Michigan23(V30K E1M)&	2708	&2681&	2663&	2851	&2893&-\\ \hline
socfb-Penn94(V42K E1M)	&3836	&3809&	3802&	4096	&4116&-\\ \hline
socfb-Texas80(V32K E1M)&	2787	&2770&	2750	&2984&	3010&-\\ \hline
socfb-Texas84(V36K E2M)&	2840&	2830&	2822&	3061&	3073&-\\ \hline
web-edu&	252&	249&	249&	251&	253&-\\ \hline
web-polblogs&	115&	109&	108&	113&	118&-\\ \hline
web-spam	&889&	858&	854&	901&	925&-\\ \hline
web-indochina-2004	&1513&	1496&	1491&	1504&	1517&-\\ \hline
web-webbase-2001&	1112&	1064&	1055	&1114&	1158&-\\ \hline
web-sk-2005&	31166&	30014&	29046&	30128&	32306&-\\ \hline
web-uk-2005&	1715&	1421	&1421&	1587	&1717&1421\\ \hline
web-arabic-2005&	19518&	18191&	17676&	18533&	20288&-\\ \hline
web-Stanford&	18398&	16924&	15001&	16155&	19678&-\\ \hline
web-it-2004&	34066&	33233&	33183&	34017&	34442&32997\\ \hline
web-italycnr-2000&	23832&	22827&	22665&	23304&	24372&-\\ \hline
\end{tabular}
\end{center}
\label{nr}
\end{footnotesize}
\end{table*}

The running time of  \cite{ali2} and Algorithm \ref{alg1} are the same since they have a small difference which does not affect the running time.
But as it can be seen the quality of solution in Algorithm \ref{alg1} is better than \cite{ali2}.  Because in \cite{ali2} they compute a total dominating set but we compute a dominating set. Essentially the size of MTDS can be twice of the size of MDS so this can explain why this happens.

The running time of algorithm \ref{alg2} depends on the algorithm used for computing MCV of $G''$.
In our experiment we have used a $2$-approximation factor algorithm for computing MVC. 
Theoretically the quality of solution in Algorithm \ref{alg2} is better than Algorithm \ref{alg1}. Because in graph $G''$ which is constructed from $G$ based  on Theorem \ref{main1}, the edges are added between the vertices which are marked in Algorithm \ref{alg1}, obviously the size of vertex cover of $G''$ is less than or equal the number of total vertices marked in Algorithm \ref{alg1}. On the other hand Algorithm \ref{alg1} is faster than Algorithm \ref{alg2}.

Note that the running time of the modified versions of algorithms (column 1 and column 2) is the same as Algorithm \ref{alg2}. In all of the instances Algorithm \ref{alg2} performs better than Mod1 and Mod2 except two instances.

The first modified version computes a total dominating set and we can compare the results with \cite{ali2} which also computes a total dominating set. As it can be expected Mod1 performs better than \cite{ali2}.

Note that \cite{cai20} is a recent sequential algorithm for computing dominating set and they have done many experiments and compared their results with state of the art algorithms. Their algorithm performs better than the other algorithms in most of the times. As we explain earlier, theoretically improving the approximation factor of MDS in distributed model is a challenging problem. If we compare our results with \cite{cai20} we can see that either our algorithms solution quality is better than their algorithm for example in BHOSLIB data set, or our solutions are at most two times of their solutions.
This shows that in practice the proposed algorithms have acceptable solutions in distributed model.

Note that we have implemented a centralized version of our proposed algorithms. In centralized version the nodes mark a vertex one by one, but in distributed version this is done by all nodes in one round. So, solution set in both centralized and sequential implementation is the same.
However we can modify the algorithms in sequential model to get better solutions which is not our aim in this paper and  we have focused on distributed algorithms.
The experiments were run in a system with OS: CentOS Linux release 7.7.1908 (Core), CPU: Intel E5-2683 v4 Broadwell $@$ 2.1Ghz
and Memory: 100G. 
The codes are also available in the web\footnote{https://github.com/salarim/MDS}.

We have run the algorithms for each instance just once. In the first step we assign a random number $r_i$ for each vertex. This can affect the solution in the case where two vertices $v_i$ and $v_j$ have a common neighbor $v$  and $d_i=d_j$ and their degree is maximum among $N(v)$. Here , $v$ marks one of them based on $r_i$ and $r_j$.
So, if we run the algorithm several times and choose the minimum solution, better results can be achieved.
 
\section{Remarks, importance and applications of proposed algorithms}

\subsection*{Set cover problem}
In the set cover problem we are given a set $A = \{a_1, a_2, \dots a_n\}$ of $n$ elements and $m$ subsets, $A_1, A_2, \dots, A_m$ of $A$. The goal is to choose the minimum number of subsets that cover all the elements of $A$. In \cite{ali2} they have explained how to change their algorithm to choose the subsets. Each element $a$ chooses a subset $A_j$ with maximum size such that $a_i\in A_j$. Let $x_i$'s be the number of times that $A_i$'s are chosen by the elements. Next round  each $a_i$ chooses a subset $A_j$ such that $a_i\in A_j$ with maximum $x_i+r_i$. For $m$ rounds the previous step is repeated.

Now we explain how to modify Algorithm \ref{alg2} to solve set cover problem. First we run the modified version of \cite{ali2}. Next we construct a graph $G'$ that its vertices are the subsets $A_1, A_2, \dots, A_m$. For each $a\in A$ we choose two subsets $A_i$ and $A_j$ with maximum values of $x_i$s such that $a\in A_i$ and $a\in A_j$ and add an edge between them.
Similar to the proof of Lemma \ref{1} It can be shown that a vertex cover for $G'$ is a set cover for $A$.

\subsection*{$k$-distance dominating set}

A $k$-observer $Ob$ of a network $N$ is a set of nodes in $N$ such that each message, that travels at least
$k$ hops in $N$, is handled (and so observed) by at least one node in $Ob$.
A $k$-observer $Ob$ of a network $N$ is minimum iff the number of nodes
in $Ob$ is less than or equal the number of nodes in every $k$-observer of
$N$ (See \cite{cha}). This problem is equivalent to the  $k$-distance dominating set problem.
 In this problem for each node $v$, the neighbors of $v$, is the set of nodes that their distance from $v$ is less than $k+1$. Then we apply the proposed algorithms as before.
 
Note that computing a minimum $k$-distance dominating set for a graph $G$ is equivalent to computing a minimum dominating set for $G^k$, where $G^k$ is a graph with the same vertex set as $G$ and we put an edge between two vertices in $G^k$ if the distance between them in $G$ is less that $k+1$. Since our algorithms performs well in dense graphs so if $G$ is a dense graph then as the value of  $k$ increases the graph $G^k$ will be denser and the quality of solution of our proposed algorithm will be improved.

\subsection*{Other variations and constraints}

Suppose that the network is dynamic and nodes and edges are added or deleted constantly for example some nodes are online only in particular period of time. In dynamic model, for example when a vertex $v$ with its adjacent edges are added we only need to change the marked vertices that are chosen by $N(v)$ and $N(N(V))$ in Algorithm \ref{alg1}. In Algorithm \ref{alg2}, only the corresponding edges of $N(v)$ and $N(N(v))$ are modified because the value of $x_i$'s for $N(v)$ is changed. This modification is done locally.

In some situations each vertex $v$ should be dominated by a specific set of vertices denoted by $A_v\subset N(v)$. In this case, in Algorithm \ref{alg1}, $v_i$ marks a vertex $v_j\in (N(v)\bigcup v) \bigcap A_v$ with maximum $w_j$. Similarly in Algorithm \ref{alg2}, for a vertex $v$ we choose two vertices with maximum $w_i$'s from  $(N(v)\bigcup v) \bigcap A_v$.
In practice, for example in a sensor network suppose the case where the coverage radius of each sensor is limited for example less than $d$.
And each sensor covers limited angular direction.

\subsection*{Remarks}

 In obtaining the upper bound, we can use other known algorithms for computing the minimum vertex cover and maximum independent set of graphs to achieve better bounds.

We believe that Algorithm \ref{alg2} is a powerful tool for computing good approximation factor solutions
for MDS and MTDS of graphs in distributed model. 
 In Algorithm \ref{alg2} or in its modified versions we try to use good candidates as our dominating set for constructing $G'$ and $G''$. So we use the output of Algorithm \ref{alg1} or  the algorithm of \cite{ali2}. As a future work one can use other algorithms or ideas to choose the vertices for adding edges between them in constructing $G'$ and $G''$.

Also it might be useful to construct the graphs $G'$ and $G''$ according to the topology of the network and in a more data sensitive way.

The important property of Algorithm \ref{alg2} is that $G'$ and $G''$ can be constructed in distributed model. The rest is computing the vertex cover of $G'$ and $G''$ which are well studied in the distributed model and we can use the known distributed algorithms for computing the minimum vertex cover. Beside the distributed nature of our proposed algorithms, it can easily seen that these algorithms can be applied on big data as well.
The proposed algorithms are very fast and easy to implement and they need low storage.

Note that the idea of constructing $G'$ and $G''$ from $G$ can help us to combine the algorithms to get a better solution. For example suppose that there are two algorithms for the MTDS(MDS). We run both algorithms and we use the solution set of both algorithms to get a better result. It is enough that in constructing the graph $G'$ ($G''$) for each vertex $v$, we choose two vertices from $N(v)$($N(v)\bigcup v$), one from the first solution set and the other from the second solution set. Then we add an edge between them. This way obviously both solution sets are a vertex cover for $G'$ ($G''$) and so the MVC of $G'$ ($G''$) is less than the size of both of solution sets.
 
\section{Conclusion}

In this paper, we presented some theoretical results for computing MDS and MTDS. We obtained an upper bound for the MDS and MTDS and gave a distributed randomized algorithm to achieve this bound.  Two distributed algorithms for computing a dominating set of a graph are presented. We implemented these algorithms and presented some experimental results to show the efficiency of our algorithms. Then we discussed the importance and applications of the proposed methods.